\documentclass[conference]{IEEEtran}
\IEEEoverridecommandlockouts

\usepackage{cite,comment}
\usepackage{amsmath,amssymb,amsfonts,amsthm}
\usepackage{algorithm}
\usepackage{algpseudocode}
\usepackage{graphicx}
\usepackage{textcomp}
\usepackage{xcolor,cancel,bm,bbm}
\usepackage[normalem]{ulem}
\usepackage{subcaption}
\newtheorem{defn}{Definition}
\newtheorem{thm}{Theorem}
\newtheorem{corollary}[thm]{Corollary}
\newtheorem{lemma}[thm]{Lemma}
\renewcommand{\S}{\mathcal{S}}
\newcommand{\F}{\mathcal{F}}
\newcommand{\G}{\mathcal{G}}
\newcommand{\M}{\mathcal{M}}
\newcommand{\Mi}{\mathcal{M}_i}
\newcommand{\indi}{\mathbbm{1}}
\newcommand{\dotsM}{1, \ldots ,M}
\newcommand{\dotsK}{1, \ldots ,K}
\newcommand{\dotsN}{1, \ldots ,N}

\def\BibTeX{{\rm B\kern-.05em{\sc i\kern-.025em b}\kern-.08em
    T\kern-.1667em\lower.7ex\hbox{E}\kern-.125emX}}
\begin{document}

\title{Handover-Optimal User Association Policy for LEO Satellite-based 5G NTN}

\author{\IEEEauthorblockN{Pradnya Taksande\IEEEauthorrefmark{1}, Jainesh Mehta\IEEEauthorrefmark{2},
Prasanna Chaporkar\IEEEauthorrefmark{3} }
\IEEEauthorblockA{Department of Electrical Engineering,
Indian Institute of Technology Bombay, India\\
Email:
\{20001816, 210071001, chaporkar\}@iitb.ac.in
}
}

\maketitle

\begin{abstract}
The integration of Non-Terrestrial Networks (NTN) into 5G and beyond cellular systems has introduced a significant paradigm shift, enabling ubiquitous connectivity and extending services to previously unconnected and underserved remote regions. In particular, Low Earth Orbit (LEO) satellites, operating close to the Earth’s surface, can provide communication latency comparable to that of terrestrial networks. However, due to their high mobility, LEO satellites trigger frequent handovers, which degrade users’ quality of experience and increase signaling overhead. In this work, our objective is to minimize the number of handovers in a LEO satellite system while preventing satellite overloading. We formulate the problem within a game-theoretic framework and apply the Spatial Adaptive Play (SAP) algorithm to obtain a handover-efficient and load-balanced solution. Additionally, we propose a low-complexity heuristic algorithm to achieve similar objectives with reduced computational overhead.
\end{abstract}


\section{Introduction}
With the evolution of 5G-Advanced and future 6G cellular networks, a key design objective is to achieve truly global coverage, service ubiquity, and reliable continuous connectivity. To meet these ambitious requirements, non-terrestrial networks (NTNs) have emerged as a promising and essential complement to traditional terrestrial networks \cite{rinaldi2020non}. NTNs leverage spaceborne and airborne platforms, such as satellites in Low Earth Orbit (LEO), Medium Earth Orbit (MEO), Geostationary Orbit (GEO), as well as High Altitude Platform Stations (HAPS) to extend connectivity beyond the limitations of ground-based deployments. By integrating NTNs with terrestrial cellular networks, it becomes possible to provide seamless connectivity to users anywhere on the globe, including remote, rural, maritime, and aerial environments.

Recognizing the potential of NTN, the Third Generation Partnership Project (3GPP) has incorporated NTN support into 5G cellular standards with significant work introduced in Release 15 and further enhancements in subsequent releases \cite{tr38821}. This standardization effort enables the harmonized operation of terrestrial and non-terrestrial components within a unified cellular framework. LEO satellites typically operate at altitudes between approximately 300 km and 1500 km above the Earth’s surface. Their close proximity to the ground significantly reduces the signal propagation delay, resulting in much lower communication latency compared to MEO and GEO satellite systems, and latency levels that are comparable to those of terrestrial networks. In addition, LEO satellites offer advantages such as high speed throughput, lower launch and manufacturing costs, reduced power needs, and improved resiliency. In this work, we consider an LEO satellite system with a 5G gNB onboard.

An LEO satellite is typically visible from a fixed location on Earth for only about 8–12 minutes, depending on its altitude. Due to the high orbital speeds of LEO satellites, ensuring continuous connectivity between users and satellites is a key challenge in such networks. As a result, users must periodically transition their connection from one satellite to another through a process known as satellite handover. Based on user-reported measurements and satellite ephemeris data, the network identifies a set of candidate target satellites. It then proactively selects the most suitable target and initiates the handover following the standard phases of the 3GPP 5G—namely handover preparation, execution, and completion. Since satellite trajectories are predictable and available in advance through ephemeris data, handover events can be pre-determined and strategically planned for a given satellite system. 

Unlike terrestrial networks, where handovers are primarily triggered by user mobility, handovers in LEO satellite networks occur mainly due to satellite mobility. The frequent movement of satellites makes handovers inevitable, even for stationary users. This leads to significant signaling overhead, increased power consumption, and potential service interruptions, which negatively affect throughput and latency. Therefore, our objective is to minimize the number of handovers in LEO systems. Reducing handovers can enhance overall system performance by reducing power consumption and improving both throughput and latency. 

In \cite{juan2022handover}, the authors give a general overview of the different solutions for handover in LEO satellite systems. In \cite{wu2019satellite}, the authors propose a handover strategy by modeling the user satellite connection as a potential game on a bipartite graph. In \cite{lee2025multi}, the authors propose a distributed multi-agent deep Q-network-based handover strategy for LEO satellite networks that enables each user to independently make real-time handover decisions to minimize handovers while maximizing throughput and Quality of Service (QoS) satisfaction. \cite{lee2025handover} puts forward a distributed Graph Neural Network (GNN)-based handover procedure for LEO satellite networks that selects target satellites while ensuring load balancing to maximize the user sum rate and reduce handover failures. 
The authors in \cite{wu2016graph} propose a graph-based customizable handover framework that jointly determines the optimal handover timing and target satellite in LEO networks by modeling satellite instances over time as a graph and selecting the QoS-optimal handover sequence using shortest path optimization. 

\cite{leng2021caching} proposes a caching-aware intelligent handover strategy for LEO satellite networks based on deep reinforcement learning, jointly considering caching capacity, remaining service time, and idle channels to reduce handover failures and call blocking. 
In \cite{he2020load}, the authors propose a distributed load-aware satellite handover strategy based on multi-agent reinforcement learning that minimizes average handovers while reducing the user blocking rate in LEO networks.
In this work, our objective is to design a handover policy for the LEO satellite system that minimizes the number of handovers while ensuring that the system is not overloaded.
Our work differs from existing studies \cite{lee2025handover,wu2016graph,leng2021caching,he2020load} in that it  focuses on minimizing the instantaneous number of handovers rather than optimizing other metrics such as the average number of handovers \cite{he2020load}, handover failures \cite{lee2025handover,leng2021caching}, or handover delay \cite{wu2016graph}.
\section{System Model and Problem Formulation}
We consider a communication system with $M$ stationary (not mobile) users and $N$ LEO satellites. Note that each satellite, depending on its position in the orbit, covers a certain limited area on the surface of the earth. Thus, based on user’s position, it may have potential connectivity to only a few satellites at a given point in time. Moreover, this connectivity evolves over time due to the mobility of the satellites. Let $c_{ij}(t)$ denote the coverage indicator between user $i$ and satellite $j$ at time $t$. $c_{ij}(t)$ is defined as follows:
\[
    c_{ij}(t)= 
\begin{cases}
    1,& \text{if user } i \text{ can associate to satellite } j \text{ at time }t\\
    0,              & \text{otherwise.}
\end{cases}
\]

Let $C(t) = [c_{ij}(t)]_{i\in\dotsM,j\in\dotsN}$. We refer to the $M\times N$ matrix $C(t)$ as the {\it coverage matrix} at time $t$ as it captures potential associations between users and satellites at time $t$. Since satellites move in a predictable manner and users are assumed to be stationary, $C(t)$ can be calculated in advance for every time $t$. Let us define $t_0 = 0$, and time instance $t_k$ recursively as follows:
\begin{align*}
t_k = \inf\{t\ | \ t>t_{k-1} \mbox{\ and } C(t) \not= C(t_{k-1})\}.
\end{align*}
Note that $\{t_k\}_{k\ge 0}$ denotes the time epochs at which the coverage matrix changes. Alternatively, the coverage matrix remains unchanged in the interval $(t_{k-1},t_{k}]$ for every $k \ge 0$. We refer to this as the $k$th time slot. Let $C[k]$ equal $C(t_{k})$. Due to the periodic nature of satellite mobility in their respective orbits, we assume that the sequence $\{C[k]\}_{k\ge 1}$ is also periodic, that is, there exists $K$ such that $C[k] = C[K+k]$ for every $k \ge 1$. Our aim is to manage the association between users and satellites over time. To achieve this, it is sufficient to make association decisions at discrete time epochs $\{t_1,\ldots,t_{K}\}$  due to the periodicity in the coverage matrix. Let $\mathcal{S}_i[k]$ denote the set of satellites with which user $i$ can associate in the $k$th time-slot. Define:
\[
\mathcal{S}_i[k] = \{j | c_{ij}[k] = 1, j = \dotsN\}.
\]

Next, we formulate our problem.

\begin{defn}[Association Policy]\label{defi:association policy}
    An association policy $\delta$ is a function that maps each user and time instant to a satellite:
    \begin{align*}
        \delta: \{\dotsM\} \times \{\dotsK\} \longrightarrow  \{\dotsN\}. 
    \end{align*}
    
\end{defn}

Let the $K$-dimensional vector $\bm{x_i^\delta} = [x_{i1}^\delta \ \cdots \ x_{iK}^\delta]$ denote the association for user $i$ under policy $\delta$. Specifically, $x_{ik}^\delta = j$ indicates that user $i$ is associated to satellite $j$ in 
$k$th time-slot under policy $\delta$.

\begin{defn}[Feasible Association Policy]\label{defi:feasible policy}
    An association policy $\delta$ is said to be  feasible if: 
    (a)~$x_{ik}^\delta \in \mathcal{S}_i[k]$ for every $i$ and $k$.
    Let $\F$ denote the set of all feasible association policies.
\end{defn}

The condition~(a) of Definition~\ref{defi:feasible policy} ensures that at each point in time a user is associated with a satellite that covers it. 

For a user $i$, the total number of handovers under $\delta$ is given by:
\begin{equation}
\label{Hi_delta}
H_i(\delta) = \sum_{k=1}^{K}\indi(x_{i(k-1)}^\delta \neq x_{ik}^\delta).
\end{equation}
The handover of the user $i$ would incur a handover cost $H_i(\delta)$. The total handover cost for the policy $\delta$ would then be:
\[
H(\delta) = \sum_{i=1}^M H_i(\delta). 
\] 

Each satellite has finite resources like bandwidth, power, etc. Hence, if a large number of users are associated with a single satellite, then the user experience may degrade. Hence, we assume that each satellite $j$ can meet the QoS requirements of all associated users if the total number of users associated with satellite $j$ is less than a given constant, say $U_j$. The value of $U_j$ may depend on the availability of resources at the satellite and the QoS requirement of the users. If the number of users associated with satellite $j$ exceeds $U_j$, then the QoS experienced by all users associated with satellite $j$ deteriorates. For user $i$ associated with satellite $j$, the overload cost is given by
\[
O_i(\delta) = 
     \indi( \sum_{m=1}^M\indi(x^{\delta}_{mk} = j) > U_j ).
\]
The total overload cost for policy $\delta$ would then be:
\[O(\delta) = \sum_{i=1}^M O_i(\delta). \]

\begin{defn}[Handover-Optimal Association Policy]\label{defi:optimum policy}
An association policy $\delta_{\lambda}^\star$ is optimal if (a)~it is feasible ($\delta_{\lambda}^\star \in \F$), and (b)~it minimizes the sum of the handover and overload cost over all feasible policies, i.e.,
\[
\delta_{\lambda}^\star \in \arg \min_{\delta \in \mathcal{F}} (H(\delta) + \lambda O(\delta)),
\]
where $\lambda$ denotes the constant that represents the overload cost associated with QoS degradation.
\end{defn}

We refer to $\delta_{\lambda}^\star$ as Handover-Optimal Association Policy (HOAP).  

\begin{lemma}\label{lemma1}
The total cost is a monotonically non-decreasing function of $\lambda$. In other words, for $\lambda > \lambda^{\prime}$, the total cost corresponding to the optimal solution with parameter $\lambda$ is greater than that with $\lambda^{\prime}$, i.e.,
\[
H(\delta_{\lambda}^{\star}) + \lambda O(\delta_{\lambda}^\star) > H(\delta_{\lambda^{\prime}}^\star) + \lambda^{\prime} O(\delta_{\lambda^{\prime}}^\star).
\] 
\end{lemma}
\section{HOAP without overload cost}
Finding HOAP, that is, $\delta_{\lambda}^\star$ is NP-hard. However, when $\lambda=0$, the overload cost is 0 and HOAP ($\delta_{0}^\star$) can be found in polynomial time, which we discuss in the next subsection.

\subsection{Optimal Handover Strategy ($\delta_0^\star$)}
When $\lambda = 0$, the overload cost in the network becomes zero, and the objective is reduced to minimizing only the handover cost. This case therefore represents the lower bound on the total cost and can be found in polynomial time. The following lemma characterizes this scenario.
\begin{lemma}\label{lemma2}
If each user follows a path where they associate with the satellite that ensures the longest continuous duration of connectivity and indulges in a handover only at the time instant when connection is lost, the total number of handovers is minimized.
\end{lemma}
\begin{proof}
We will analyze the simple case of one user and multiple satellites to show that if each user follows the path in which they connect to the satellite that ensures the longest continuous duration of connectivity, the number of handovers for that user will be minimized. Since the total number of handovers is the sum of individual handovers, if it is ensured that each user follows this strategy, then the total number of handovers will also be minimized.

Given a single user and $n$ satellites, our objective is to minimize the number of handovers that occur in $K$ time instants. At each time instant the user can connect with only a subset of satellites, as specified by the matrix $C(t)$. The element $c_{ij}(t)$ indicates whether the user $i$ can connect to satellite $j$ at time instant $t$ or not. Our task is to determine the assignment policy $\delta$ for each $t$ which gives us $x_i^\delta(t)$,  where $x_i^\delta(t) = j$ if the user $i$ is associated to satellite $j$ at time $t$. The number of handovers for user $i$ would be defined by eq. (\ref{Hi_delta}).
The greedy strategy to solve the problem is to assign the user to the satellite that provides the longest duration of connectivity to the user before switching to another satellite. 
Specifically, if the user $i$ is associated to a satellite $j$ at any time instant $t$ then it should remain associated to that satellite until connectivity is lost, that is, when $c_{ij}(t)$ becomes 0. At the time instant when connectivity is lost, the user shifts to a satellite $j^\star$ which provides the longest duration of uninterrupted connectivity to the user i.e.,
\[j^\star = \arg\max_{j \in \S(t)}d_j(t),\]
where $\S(t)$ is the set of satellites available at time $t$ and $d_j(t)$ represents the time duration for which connectivity of the satellite can be ensured and is defined as:
\[d_j(t) = \max\{\tau \geq 0 | c_{ij}(t+k) = 1 \: \forall \: k \in \{0,1,\dots, \tau -1\}\}. \]

This greedy algorithm ensures that the user is associated to a satellite for the longest possible duration before switching, thereby reducing the number of handovers.
To establish that the greedy approach is optimal we consider any alternate assignment where at some time instant $t$, the user is switched to some other satellite $j'$ that does not provide the longest continuous connectivity among all available satellites; then the user will surely experience a handover faster than that experienced in the greedy approach. This introduces an additional handover, which contradicts the assumption that the assignment is optimal. Each decision in the greedy strategy delays each handover as much as possible and thus minimizes the total number of handovers in the entire time horizon. This is the optimal policy $\delta_0^\star$ where the overload cost is zero.
\end{proof}

Next, we consider the case $\lambda > 0$, where the total cost is the sum of the handover cost and the overload cost. In this scenario, if each user independently follows its optimal path ($\delta_0^\star$) and associates with the satellite that provides the longest continuous connectivity duration, the resulting assignment would be optimal from an individual user’s perspective. However, this strategy may lead to satellite overloading, since multiple users may have overlapping paths and may prefer to associate with the same satellite at a given time instant. 

When a satellite is overloaded, each  user connected with the satellite faces a degradation in QoS which is the overload cost per user. In such situations, certain users must deviate from their optimal assignment to balance the overall handover and overload costs. Unlike the single-user case, the multi-user scenario introduces interdependencies, where the decision of one user impacts the others. 

The assignment of users to satellites must be determined globally by considering both the current and future time instants, due to the interdependencies among users. This makes the problem NP-hard. However, the user dependency is largely localized, since users that are geographically close tend to compete for the same satellite resources. Therefore, the problem can be simplified to resource sharing within local groups of users rather than across all users in the network. This allows the problem structure to be modeled as a local interaction game. In the following section, we present the formulation of this local interaction game.

\section{Local Interaction game}
We integrate our model into the game-based analytical framework. We model each user $i$ as a player who chooses an action $x_{ik}$ at $k$th time-slot. $x_{ik}$ takes values from respective finite action set $\mathcal{S}_i[k]$, consisting of the satellites to which the user $i$ can associate at time-slot $k$. Then the strategy of user $i$ for $K$ time-slots is the vector $\mathbf{x_i}:= (x_{i1}, x_{i2}, \dots x_{iK})$ and let $\mathcal{S}_i = \prod_{k=1}^{K} \mathcal{S}_i[k]$ denote the collection of all possible strategies (states) of player $i$. Then $\mathbf{x_i} \in \mathcal{S}_i$. The terms user and player are used interchangeably throughout the rest of the paper. 

There are a total of $M$ players in the network. Let $\M = \{1,\ldots,M\}$ denote the set of the players in the network. Let the vector $\mathbf{x}:= (\mathbf{x_1}, \mathbf{x_2}, \dots ,\mathbf{x_M})$ denote the state (or the strategy profile) of all players. Let us denote $\mathcal{S} := \mathcal{S}_1 \times \mathcal{S}_2 \dots \times \mathcal{S}_M$. Then $\mathbf{x} \in \S$. Let the utility function of player $i$, which is dependent on the actions of all players, be denoted as $u_i(\mathbf{x}) : \mathcal{S} \rightarrow \mathbb{R}$. Then $\G = $ \{$ {\M, \mathcal{S}, (u_i(\mathbf{x}),i \in \M)}$\} is a game. 

\subsection{Definition of utility function}
\textit{Definition of $u_i(\mathbf{x})$:}
We now begin to define the utility function $u_i(\mathbf{x}),i \in \M$.
The handover cost experienced by player $i$ is defined as:
\[
H_i(\mathbf{x_i}) = \sum_{k=1}^{K}\indi(x_{i(k-1)} \neq x_{ik}).
\]
Now, the total handover cost of all the players is given by:
\[
H(\mathbf{x}) = \sum_{i=1}^{M} H_i(\mathbf{x_i}).
\]

Now, let $Y_{jk}(\mathbf{x}) = \sum_{i=1}^M\indi(x_{ik} = j)$ denote the number of users associated with satellite $j$ in the time-slot $k$. Let $U_j$ indicate the maximum number of users that can associate to a satellite $j$ at any point of time so that their QoS are satisfied. If user $i$ overloads satellite $j$, that is, if user $i$ associates to satellite $j$ and $j$ already has $Y_{jk}(\mathbf{x}) = U_j$ users associated with it, then satellite $j$ is overloaded and user $i$ incurs an overload cost of $\lambda$. This overload cost for player $i$ is defined as:
\[
O_i(\mathbf{x}) = 
\begin{cases}
    \begin{aligned}
    0, & \quad \text{ if } Y_{jk}(\mathbf{x}) \leq U_j , \\
    \lambda , & \quad \text{ otherwise,}
    \end{aligned}
\end{cases}
\]
where $\lambda$ is a constant denoting per-user overload cost. 

Then, the total overload cost in the network is given by:
\[
O(\mathbf{x}) = \sum_{i=1}^{M} O_i(\mathbf{x}).
\]

Our objective is to determine user associations that minimize both handover and overload costs across the network. The optimization problem is as follows:
\begin{equation}
\label{p1}
\begin{aligned}
\mathcal{P}: \min_{\mathbf{x} \in \S} \quad & \sum_{i \in \M} H_i(\mathbf{x}) + O_i(\mathbf{x}).
\end{aligned}
\end{equation}

Now, construct the welfare function $f_i(\mathbf{x})$ corresponding to every player $i$ as
\begin{equation}
\label{fi}
    f_i(\mathbf{x}) = - H_i(\mathbf{x}) - O_i(\mathbf{x}).
\end{equation}

Thus, let the social welfare function be defined as 
\begin{equation}
    \label{f}
f(\mathbf{x}) := \sum_{i \in \M} f_i(\mathbf{x}).
\end{equation}

Next, we formally define the \textit{Neighbors} of a player.

\begin{defn}[Neighbors]\label{defi:neighbors}
Users $i$ and $i'$ are neighbors if $\exists k$ such that $\mathcal{S}_i[k] \cap \mathcal{S}_{i^{\prime}}[k] \neq \emptyset$. Users are considered neighbors if they can connect to the same satellite in a time-slot $k$. Let 
\[
\Mi = \left\{i^{\prime}: \cup_{k=1}^{K} ( \mathcal{S}_i[k] \cap \mathcal{S}_{i^{\prime}}[k] ) \neq \emptyset \right\}.
\] 
\end{defn}
Note that nodes $i$ and $i'$ will be neighbors if they are geographically close to each other. Thus, $|\Mi| \ll |\M|$, that is, the number of neighbors of a user is significantly smaller than the total number of users in the system. In other words, each user interacts only with a small subset of users, resulting in a sparse interaction graph. By this definition of Neighbors, the nodes are symmetric, that is, if node $i$ is a neighbor of $i'$, then node $i'$ is a neighbor of $i$ and the interaction graph simplifies to an undirected graph.

Now, define the utility function of a user $i$ as the sum of the welfare functions of all its neighbors.
\begin{equation}
\label{utility}
u_i(\mathbf{x}) = \sum_{m \in \Mi} f_m(\mathbf{x}),
\end{equation}
where $\Mi$ denotes the neighbors of user $i$. That is, the utility function of the user $i$ depends only on the welfare functions of its neighbors. The undirected graph formed by $(\M,\{\Mi\})$ is called as the interaction graph with each player in $\M$ represented as node and $\Mi$ specifies their neighbors. This implies that if player $i$ changes its decision unilaterally, it affects only the utilities of the players in $\Mi$, while the utilities of other players remain unchanged. Therefore, the action of a player leads to only local changes in utility within the interaction graph. Hence, it is called a local interaction game. Thus, our game $\G = $ \{$ {\M, \mathcal{S}, (u_i(\mathbf{x}),i \in \M)}$\} is a local interaction game.
%
%
\subsection{Exact potential game}
Consider the following notation. For $\mathbf{x} = (\mathbf{x_i}, i\in \M) $, denote $\mathbf{x}_{-i} := (\mathbf{x_1},\ldots,\mathbf{x_{i-1}},\mathbf{x_{i+1}},\ldots,\mathbf{x_M})$ and $(\mathbf{x}'_i,\mathbf{x}_{-i}) := (\mathbf{x_1},\ldots,\mathbf{x_{i-1}},\mathbf{x}'_i,\mathbf{x_{i+1}},\ldots,\mathbf{x_M})$. 
%
\begin{defn}[Exact Potential game]\label{defi:potentialgame}
A game \{${\M, \mathcal{S}, (u_i,i \in \M)}$\} is said to be an exact potential game if there exists a function $V : \mathcal{S} \rightarrow \mathbb{R}$, known as a potential function that satisfies
\begin{equation} \label{pot_func}
u_i(\mathbf{x}'_i,\mathbf{x}_{-i}) -u_i(\mathbf{x}) = V(\mathbf{x}'_i,\mathbf{x}_{-i}) - V(\mathbf{x}),
\end{equation}
for all $i \in \M, \mathbf{x}'_i \in \mathcal{S}_i, \mathbf{x} \in \S$.
\end{defn}
In a potential game, the incentives for all players to change their strategy are captured by a single global function called the potential function $V$. All maximizers of a potential function $V$ are Nash equilibria of the potential game. If the local utility of users is maximized, the global social welfare of the system would, in turn, be maximized. The utility functions of players are totally aligned with our objective of maximizing social welfare, which leads to the following lemma.
\begin{lemma}\label{lemma3}
The local interaction game $\G$ is an exact potential game, with the social welfare function $f(\mathbf{x})$ as its potential function $V(\mathbf{x})$.
\end{lemma}
\begin{proof}
Let the players use a strategy $\mathbf{x} = (\mathbf{x}_i, i \in \M)$. Consider a player $i$ and assume that it changes its strategy from $\mathbf{x}_i$ to $\mathbf{x}'_i$. Here, the potential function $V(\cdot) = f(\cdot)$. Then, to prove this theorem, we need to prove the condition in Eq. (\ref{pot_func}) holds. Now,
\begin{align*}
&V(\mathbf{x}'_i,\mathbf{x}_{-i}) - V(\mathbf{x})\\
&= f(\mathbf{x}'_i,\mathbf{x}_{-i}) - f(\mathbf{x})\\
&= \!\!\!\sum_{m \in \M} (f_m(\mathbf{x}'_i,\mathbf{x}_{-i}) - f_m(\mathbf{x}) )\\
&= \!\!\!\sum_{m \in \Mi} (f_m(\mathbf{x}'_i,\mathbf{x}_{-i}) - f_m(\mathbf{x}) ) + \!\!\!\sum_{m \notin \Mi} (f_m(\mathbf{x}'_i,\mathbf{x}_{-i}) - f_m(\mathbf{x}) )\\
&\stackrel{(*)}{=} \!\!\!\sum_{m \in \Mi} f_m(\mathbf{x}'_i,\mathbf{x}_{-i}) - \!\!\!\sum_{m \in \Mi} f_m(\mathbf{x}) )\\
&= u_i(\mathbf{x}'_i,\mathbf{x}_{-i}) - u_i(\mathbf{x}),
\end{align*}
where $\stackrel{(*)}{}$ follows from the fact that for all $m,i \notin \Mi $, their utility functions $u_m()$ are independent of $i$'s strategy. So, the function $f(\mathbf{x})$ is the potential function of the game $\G$. The game is thus an exact potential game. This theorem implies that the global social welfare function captures the effect of unilateral action changes by an individual player. 
\end{proof}

\begin{corollary}
$\exists$ at least one pure strategy Nash Equilibrium (NE). Every pure strategy NE is a local maxima for the objective (social welfare function) $f(\mathbf{x})$.
\end{corollary}

\begin{corollary}
The global maxima is also a pure strategy NE of the game $\G$. 
The strategy $\mathbf{x}^\star$ selected by the policy $\delta_{\lambda}^\star$ is a pure strategy NE of the local interaction game $\G$. 
\end{corollary}
Next, we look at algorithms to find the global optimizer of the potential function $f(\mathbf{x})$.

\section{Proposed Algorithms}
Our aim is to find the solution to the optimization $\mathcal{P}$ defined in Eq. (\ref{p1}). Since the game formulation of this problem $\G$ is an exact potential game, global optima exist and can be found using certain learning algorithms. Spatial Adaptive Play (SAP) is one such learning algorithm which finds the global optima after performing a sufficiently large number of iterations.
\subsection{SAP algorithm}
SAP is one of the algorithms that leads to an optimizer of the potential function with high probability, in the case of potential games. It is often used in spatial or networked environments where agents are connected through a graph structure, and their decisions influence each other through local interactions.
The SAP Algorithm is presented in Algorithm \ref{sap}.
\begin{algorithm}[htbp]
	\caption{SAP algorithm} 
        \label{sap}
        \hspace*{\algorithmicindent} \textbf{Input:} $\M, \{u_i(\mathbf{x})\},\beta, max_{iter} $\\
        \hspace*{\algorithmicindent} \textbf{Output:} $\mathbf{x} \gets$ global optimizer of $f(\cdot)$
	\begin{algorithmic}[1]
            \State $iter \gets$ 1
            \For {$iter < max_{iter}$}
            \State $i \gets$ uniformly random player from $\M$
            \State $\mathbf{x_i} \gets$ uniformly random strategy from the set $\mathcal{S}_i$
            \State Compute $u_i(\mathbf{x_i},\mathbf{x}_{-i})$, 
            \State $\mathbf{x_i} \gets $ sample according to probability distribution 
            \State $p(\mathbf{x_i}) = \frac{e^{\beta u_i(\mathbf{x_i},\mathbf{x}_{-i})}}{\sum_{\mathbf{z} \in \mathcal{S}_i} e^{\beta u_i(\mathbf{z},\mathbf{x}_{-i})}}$
            \EndFor
	\end{algorithmic}
\end{algorithm}
The main idea of the algorithm is to choose a strategy $\mathbf{x_i}$ for a player $i$ chosen at random and compute its utility function $u_i(\mathbf{x_i},\mathbf{x}_{-i})$ (Step 5), by changing the strategy of the player $i$ and keeping the strategies of other players intact. Then compute the following probability (Step 7):
\begin{equation}
\label{gibbs}
p(\mathbf{x_{i}}) = \frac{e^{\beta u_i(\mathbf{x_{i}},\mathbf{x}_{-i})}}{\sum_{\mathbf{z} \in \mathcal{S}_i} e^{\beta u_i(\mathbf{z},\mathbf{x}_{-i})}},
\end{equation}
Then change the strategy of the player $i$ according to this $p(\mathbf{x_{i}})$, otherwise with probability $1-p(\mathbf{x_{i}})$, do not change the strategy of the player $i$ (Step 6). We repeat this process until the maximum number of iterations $max_{iter}$ is reached. $\beta$ is the learning rate for the SAP algorithm; it determines the trade-off between exploration and exploitation in SAP. The SAP algorithm provides the following analytical guarantee.
\begin{thm} 
[As shown in \cite{kumar}]
The SAP algorithm converges to a global optimal solution with high probability. 
\end{thm}

The SAP Algorithm updates one node at a time, which can lead to slow convergence. Next, we discuss a modified version of SAP, called Concurrent-SAP (C-SAP), which is found to be effective in the case of local interaction games.
\subsection{Concurrent-SAP algorithm}
The nodes that are separated by more than two hops (i.e., non-neighboring nodes) can update their utility functions simultaneously and independently without causing conflicts. Due to this parallel update capability, this version of SAP is referred to as the Concurrent-SAP (C-SAP) algorithm. For this modification of SAP, the stationary distribution is the same as the Gibbs distribution (Eq. \ref{gibbs}). Moreover, it results in faster convergence compared to SAP, where only one node is updated at a time. For C-SAP, we use the selection scheme presented in \cite{kumar}. However, it is computationally intensive. Therefore, in the next subsection, we propose a heuristic algorithm that has a lower computational complexity while achieving performance comparable to that of the C-SAP algorithm, as demonstrated by our simulations in Section \ref{sim}.
%
%
\subsection{Heuristic algorithm}
In this subsection, we propose a greedy hill-climbing scheme that begins with an initial policy and iteratively moves toward the optimal solution by modifying one action at a time. Starting from the initial random policy, the algorithm examines each action in each iteration, changes one action, and evaluates the resulting social welfare function. If an improvement is observed, the new action is adopted. In this manner, the scheme greedily climbs toward the optimum. However, its performance depends on the choice of the initial policy. We refer to this as the greedy algorithm.

\section{Simulation Results}\label{sim}
In this section, we compare the three schemes discussed in the preceding section through extensive simulations. We consider a square region of 1000 km $\times$ 1000 km within which 74 satellites are in motion. The region is divided into four vertical bands. In each band, two satellites have a fixed position along the x-axis while their starting positions along the y-axis are randomly assigned. The remaining satellites are randomly distributed throughout the region. The velocity of each satellite is randomly selected in the range of 7 to 8 km/s. The users are assumed to be stationary and are uniformly distributed throughout the region.

The simulation parameters used are as listed in Table \ref{sim_par}, unless otherwise specified.
%
\begin{table}[ht]
\caption{Simulation Parameters}
\label{sim_par}
\begin{center}
	\begin{tabular}{|l|l|}
		\hline
		\textbf{Parameter} & \textbf{Value}\\
		\hline
		  Per user overload cost ($\lambda$) & 10\\
		\hline
		  Number of Satellites ($N$) & 74 \\
		\hline        
		Capacity of satellite ($k$) & 3 \\
        \hline
		  Number of time-slots ($T$) & 100 \\          
		\hline
        Parameter for C-SAP ($\beta$)&  5  \\
        \hline
		Parameter for C-SAP ($max_{iter}$) & 10000000 \\
		\hline        
	\end{tabular}
\end{center}
\end{table}
%
%
%
\begin{figure}[!htbp]
\centerline{\includegraphics[width=0.4\textwidth]{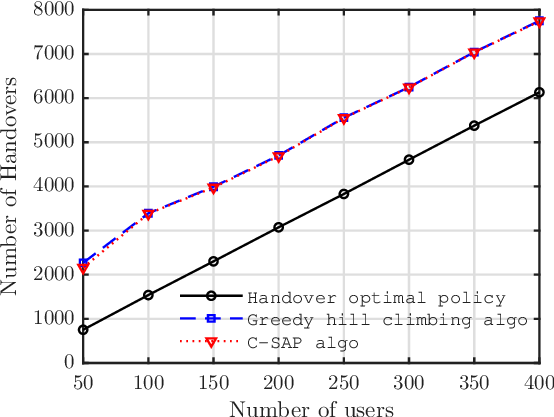}}
\caption{Number of handovers vs number of users.}
\label{HOs_users_varying}
\end{figure}
%
%
\begin{figure}[!htbp]
\centerline{\includegraphics[width=0.4\textwidth]{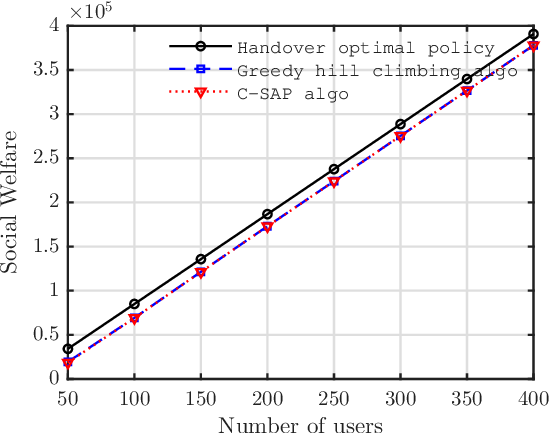}}
\caption{Social Welfare or Total Cost vs number of users.}
\label{SW_users_varying}
\end{figure}
%

%
%

Figures \ref{HOs_users_varying} and \ref{SW_users_varying} illustrate the variation in the number of handovers and total cost, respectively, as the number of users increases from 50 to 400. Green line represents the first baseline, namely the handover optimal policy ($\delta_0^\star$), which results in the minimum number of handovers, since overload cost is zero. Blue line represents the second baseline, the C-SAP algorithm, which provides the optimal policy by considering both handover cost and overload cost. However, this algorithm is computationally expensive. Therefore, a greedy heuristic algorithm (red) is proposed, which achieves performance comparable to the C-SAP algorithm.

In these simulations, the system parameters are fixed as follows: $\lambda=10$, the number of satellites is 74, the number of time-slots is 100, and each satellite has a fixed capacity of 3 users per time-slot. As the user population grows, the network becomes more congested, leading to increased handovers and higher overall system cost, showing the general trend in the plots. 

For the 'Handover optimal policy', the objective reduces to minimizing only the number of handovers. Consequently, this policy achieves the minimum number of handovers among all schemes. However, because it does not penalize satellite overloading, it allows capacity violations, resulting in artificially high social welfare in Figure \ref{SW_users_varying}. The C-SAP algorithm, which accounts for the overload cost ($\lambda = 10$), yields the lowest social welfare among the compared schemes. This is because it strictly considers both handover and overload costs, thereby avoiding excessive satellite congestion. The greedy hill-climbing algorithm is a heuristic approach that makes locally optimal decisions. Its performance closely tracks that of the C-SAP algorithm across the entire user range, indicating that it provides a reasonable approximation with lower computational complexity.
\section{Conclusions}
In this paper, we model the problem of user handovers and satellite overloading in a LEO satellite system as a local interaction game within a game-theoretic framework. To achieve a handover-optimal and load-balanced network state, we apply the SAP algorithm from game theory. To improve convergence speed, we further propose a modified version of SAP, termed Concurrent-SAP (C-SAP), in which multiple users update their utility functions simultaneously, enabling faster convergence to a stable solution. In addition, we design a simple greedy heuristic algorithm with lower computational complexity compared to SAP. Finally, we evaluated and compared the performance of these algorithms against a minimum-handover policy through extensive simulations.


\bibliographystyle{IEEEtran}
\bibliography{references.bib}

\end{document}